\documentclass[11pt]{article}
\usepackage{amsmath}
\usepackage{amsfonts}
\usepackage{amssymb}
\usepackage{amsthm}
\usepackage{bm}
\PassOptionsToPackage{table}{xcolor}
\usepackage{tikz}
\usepackage[affil-it]{authblk}
\usepackage{natbib}
\usepackage{booktabs}
\usepackage{rotating}
\usepackage{multirow}
\usepackage{array}
 
\newtheorem{teo}{Theorem}
\newtheorem{propi}[teo]{Proposition}
\newtheorem{lema}[teo]{Lemma}
\newtheorem{coro}[teo]{Corollary}

\newtheorem{defi}{Definition}
\newtheorem{axiom}{Axiom}

\newtheorem{myex}{Example} 
\newenvironment{exa}{\begin{myex}\rm}{\hfill$\vartriangle$\end{myex}}

\newcommand{\argmax}{\operatornamewithlimits{argmax}}

\renewcommand{\leq}{\leqslant}
\renewcommand{\geq}{\geqslant}
\newcommand{\NB}{\mathit{NB}}
\newcommand{\NT}{\mathit{NT}}
\newcommand{\pref}{\succ}
\newcommand{\npref}{\nsucc}
\newcommand{\prof}[1]{\boldsymbol{#1}}
\newcommand{\yes}{$\checkmark$}

\title{\bf Voting with Partial Orders: The Plurality and Anti-Plurality Classes}
\author[1]{Ulle Endriss}
\author[2]{Federico Fioravanti\footnote{federico.fioravanti@univ-st-etienne.fr (corresponding author)}\ }
\affil[1]{Institute for Logic, Language and Computation (ILLC), University of Amsterdam, The Netherlands}
\affil[2]{Universit\'e Jean Monnet Saint-\'Etienne, CNRS, Universit\'e Lyon 2, emlyon business school, GATE, 42023, Saint-\'Etienne, France}
\date{}

\begin{document}
\maketitle

%%%%%%%%%%%%%%%%%%%%%%%%%%%%%%%%%%%%%%%%%%%%%%%%%%%%%%%%%%%%%%%%%%%%%%%%%%%%%%%%

\begin{abstract}\noindent
In the theory of voting, the Plurality rule for preferences that come in the form of linear orders selects the alternatives most frequently appearing in the first position of those orders, while the Anti-Plurality rule selects the alternatives least often occurring in the final position. We explore extensions of these rules to preferences that are partial orders, offering axiomatic characterisations for them.\\[5pt]
\end{abstract}

%%%%%%%%%%%%%%%%%%%%%%%%%%%%%%%%%%%%%%%%%%%%%%%%%%%%%%%%%%%%%%%%%%%%%%%%%%%%%%%%
\section{Introduction}\label{sec:intro}
%%%%%%%%%%%%%%%%%%%%%%%%%%%%%%%%%%%%%%%%%%%%%%%%%%%%%%%%%%%%%%%%%%%%%%%%%%%%%%%%

% SCENARIO: PARTIAL ORDERS 
We consider the problem of a group choosing alternatives from a fixed set of finitely many alternatives. 
In a deviation from the model most commonly studied in social choice theory \citep{ArrowEtAlHBSCW2002,ZwickerHBCOMSOC2016}, we focus on situations in which each individual’s preference is a (strict) \emph{partial order}, allowing any given individual, for each pair of alternatives, to report either that she strictly prefers one over the other or that she does not wish to compare the two.
The study of collective decision making on the basis of partial-order preferences deserves attention, as requiring voters to fully rank a (possibly) large number of alternatives may be overly demanding in many circumstances.
Partial orders offer a less restrictive setting, giving voters more flexibility for expressing their true beliefs, while still allowing for the option to express preferences as linear orders for those who wish to do so.

\begin{exa}\label{ex:intro}
Consider the following situation, inspired by an example first discussed by \citet{TerzopoulouEndrissIJCAI2019}. 
Suppose we want to run a competition to look for the best mobile app. We ask users to rank different mobile apps, namely {\em Instagram, Facebook, TikTok, Uber, Gmail}, and {\em Yahoo}.
A particular user might be sure that she prefers {\em Instagram} to {\em Facebook}, and {\em Facebook} to {\em TikTok}, as she uses the three of them for posting videos; she also is sure that she prefers {\em Gmail} to {\em Yahoo}; but she might not be able to compare {\em Uber} to the rest of the apps, or {\em Gmail} to {\em Instagram}. 
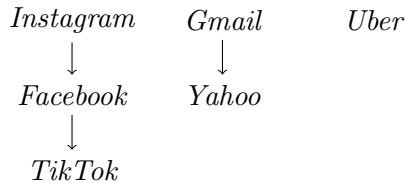
\begin{center}
\begin{tikzpicture} 
\node (1) {\it Instagram}; 
\node (5) [below of=1] {\it Facebook};
\node (4) [below of=5] {\it TikTok};
\draw[->] (1) -- (5);
\draw[->] (5) -- (4);
\node (6) [right of=1] {};
\node (7) [right of=6] {\it Gmail};
\node (8) [below of=7] {\it Yahoo};
\node (9) [right of=7] {};
\node (10) [right of=9] {\it Uber};
\draw[->] (7) -- (8);
\end{tikzpicture}
\end{center}
So it is clear that {\em Instagram} and {\em Gmail} are among the most preferred apps, and that {\em TikTok} and {\em Yahoo} are the least preferred ones.
But it is less clear whether {\em Uber} should be considered a good or a bad app, as it is not preferred to other alternatives, but there also are no alternatives that are preferred to it. 
Overall, it is not obvious which app should be considered \emph{the} best for this user, as the set of most preferred alternatives has a particular structure (e.g., {\em Instagram} beats two other apps, while {\em Uber} has not been compared to any other app).
What is even less clear is how the preferences of multiple users should be aggregated in a situation like this.
\end{exa}

\noindent
% FOCUS: PLURALITY-LIKE RULES
For the standard model of social choice, in which voters report (strict) linear orders, the most widely used rule in practice is the Plurality rule, which selects the alternative(s) most frequently ranked at the top of an individual preference order. 
While the Plurality rule has been rightfully criticised in the literature for its shortcomings \citep{Laslier2011}, it does have the significant advantage of being particularly simple: every voter simply awards a point to her most preferred alternative. 
Can we preserve some of this simplicity when moving from linear to partial orders? The answer is not obvious, as an individual who only reports a partial order might not have a uniquely identified most preferred alternative. 

If we stick to the idea that an individual's vote should contribute more to the selection of her top alternatives than to the selection of others, we arrive at what we shall call the \emph{Plurality Class} of voting rules for partial orders. 
For instance, in the mobile app example, \emph{Instagram}, \emph{Gmail}, and \emph{Uber} are all top alternatives (in the sense that there are no alternatives more preferred to them). 
One seemingly natural rule might assign 1 point to each of them (and 0 to all others), while another might assign each top alternative as many points as the number of alternatives it dominates in that preference (and still 0 to all non-top alternatives).
Both of these rules belong to the Plurality Class, and both reduce to the standard Plurality rule for the special case of linear orders.

The core contribution of this paper is the identification and axiomatic characterisation of the Plurality Class and the closely related Anti-Plurality Class, which extends the well-known Anti-Plurality rule (also referred to as the Veto rule) for linear orders. 
In the standard framework, the Anti-Plurality rule selects the alternative(s) that are least frequently ranked last. 
We further provide characterisations of several important subclasses, including, in particular, the most natural interpretations of the Plurality and Anti-Plurality rules within this setting.

% RELATED WORK: CLASSICAL
To characterise these two classes, we adapt many previously introduced axioms to our setting, while also proposing new ones.
The standard Plurality and Anti-Plurality rules belong to the family of positional scoring rules, which are rules that assign scores to alternatives based on their positions in the linear orders, and then choose the alternative(s) that maximise the sum of the scores across all voters.
Since \citet{young1975} characterised the scoring rules as the only ones satisfying Anonymity, Neutrality, Reinforcement, and Continuity, many different characterisations of the Plurality rule and the Anti-Plurality rule for linear orders have been obtained.
Let us briefly mention some of the characterisations of the Plurality rule that can be found in the literature:\footnote{Definitions of these axioms can be found in Section~\ref{sec:model} and in Section~\ref{sec:results}, except for Reduction, which says that Pareto-dominated alternatives can be deleted; Monotonicity, which requires that if a voter improves the position of an already winning alternative, keeping everything else fixed, then that alternative becomes the only winner; Bottom- and Top-Invariance, which imply that reshuffling losing alternatives below or above, respectively, the winning alternative, does not change the outcome; and Minimal Veto, which says that every minority group should be assigned veto power.} Anonymity, Neutrality, Reinforcement, Continuity, and Reduction \citep{richelsonplurality1978};\footnote{\citet{ching1996simple} later showed that Continuity, in fact, is not necessary for this characterisation.} Anonymity, Neutrality, Reinforcement, Monotonicity, and Bottom-Invariance \citep{Merlin1999ImplementationOS}; and Anonymity, Neutrality, Reinforcement, Faithfulness, and Tops-Only \citep{sekiguchi2012characterization}.
Moreover, \citet{llamazares2015scoring} show that within the scoring rules that award one point to the first $k$ alternatives, Plurality is the only one that respects the Pareto principles. 
For the case of the Anti-Plurality rule, we can mention the following characterisations for linear orders: Anonymity, Neutrality, Reinforcement, Monotonicity, and Top-Invariance \citep{BARBERA198249}; Anonymity, Neutrality, Reinforcement, Continuity, and Minimal Veto \citep{baharad2005inverse}; and Anonymity, Neutrality, Reinforcement, Averseness, and Bottoms-Only \citep{KURIHARA2018110}.

% RELATED WORK: PARTIAL ORDERS
While we are not aware of any attempts to characterise Plurality-like rules for partial orders, there is a growing literature that recognises the practical significance of partial orders in the context of social choice \citep{KonczakLangMPREF2005,TerzopoulouPhD2021}. 
For instance, \citet{pinipopref2005} generalise both Arrow's Theorem \citeyearpar{Arrow1951} and the Muller-Satterthwaite Theorem \citeyearpar{MULLER1977412} to the case of partial orders. \citet{TerzopoulouEndrissIJCAI2019} consider the problem of aggregating incomplete pairwise preferences where the voters have a weight that depends on the number of alternatives they rank.
\citet{cullinan2014borda} characterise a particularly natural adaptation of the Borda rule to election scenarios where preferences are partial orders.
Some work has focused on specific classes of partial orders, notably the top-truncated orders, where an individual strictly ranks her most preferred alternatives and it is assumed that she is indifferent between the remaining alternatives.
In this setting, \citet{BaumeisterEtAlAAMAS2012} provide complexity results on manipulation and bribery, while \citet{terzopoulou2021borda} characterise different variants of the Borda rule.

% RELATED WORK: APPROVAL
A related framework from which we draw some inspiration, and to which some of our results can be applied, is the setting of approval voting \citep{bramsfishburnapporvalvoting1978}, where ballots merely indicate which alternatives a voter approves of. 
This framework has been the subject of extensive axiomatic analysis, with numerous characterisations of various approval voting rules available in the literature.
For instance, \citet{BrandlPetersJET2022} provide eight distinct characterisations of the standard approval voting rule, while \citet{ALCALDEUNZU20091187} and \citet{BardalEndrissGEB2025} focus on approval-based scoring rules, where a voter's weight depends on the size of the approval set they report. \citet{aragones2011making} show that, in certain environments, approval voting provides stronger incentives for participation than plurality voting.

% PAPER STRUCTURE
The remainder of this paper is organised as follows.
Section~\ref{sec:model} presents the model of voting we study, introduces a number of voting rules for this setting, and defines and motivates the axioms we shall use.
Section~\ref{sec:results} presents our characterisation results for the two families of voting rules we focus on, inspired by the standard Plurality and Anti-Plurality rule. It also comments on a few applications of our results to approval voting.
Section~\ref{sec:conclusion} concludes.

%%%%%%%%%%%%%%%%%%%%%%%%%%%%%%%%%%%%%%%%%%%%%%%%%%%%%%%%%%%%%%%%%%%%%%%%%%%%%%%%
\section{The Model}\label{sec:model}
%%%%%%%%%%%%%%%%%%%%%%%%%%%%%%%%%%%%%%%%%%%%%%%%%%%%%%%%%%%%%%%%%%%%%%%%%%%%%%%%

In this section, we present the model of voting with partial orders that we study. 
It is a model of \emph{variable electorates}, meaning that any finite subset of an infinite universe of voters might cast a ballot in any given election. 
Our model is one of the standard models of voting commonly studied in the literature \citep{ArrowEtAlHBSCW2002}, except that individual preferences are assumed to be partial orders rather than linear orders.

We also recall the definition of the family of positional scoring rules for this model and then introduce the Plurality Class and the Anti-Plurality Class of voting rules belonging to this family. Finally, we propose a number of axioms for voting with partial orders.

\subsection{Voting with Partial Orders}

% ALTERNATIVES AND PREFERENCES
Let $A$ be a finite set of \emph{alternatives}, with $|A|=m\geq 3$.
A \emph{preference} is a (strict) partial order~$\pref$ over the set of alternatives, i.e., an irreflexive, antisymmetric, and transitive (but not necessarily complete) binary relation on~$A$. We use $\mathcal{D}$ to denote the \emph{domain} of all such preferences.
For $a,b\in A$, we say that $a$ is (strictly) \emph{preferred} to $b$ in case $a\pref b$. 
If neither $a\pref b$ nor $b\pref a$ holds, we say that $a$ is \emph{incomparable} to $b$.
We use $\pref^{ab}$ to denote the preference $\pref$ where alternative $a$ is relabelled as $b$ and alternative $b$ is relabelled as $a$, while everything else remains the same.
The \emph{top} of a preference $\pref$ is the set of alternatives not dispreferred to any other alternative: 
\begin{eqnarray*}
T(\pref) & = & \{\, a\in A\mid x \npref a\text{ for all }x\in A \,\}.
\end{eqnarray*}
The {\em non-top} of a preference $\pref$ is the set $\NT(\pref)=A\setminus T(\pref)$.
Analogously, the {\em bottom} of $\pref$ is the set of alternatives not preferred to any other alternative:
\begin{eqnarray*}
B(\pref) & = & \{\, b\in A\mid b\npref x\text{ for all }x\in A \,\}.
\end{eqnarray*}
The {\em non-bottom} of $\pref$ is the set $\NB(\pref)=A\setminus B(\pref)$.
Note that it is not necessarily the case that $T(\pref)\cap B(\pref)=\emptyset$. 
For instance, in Example~\ref{ex:intro}, alternative \textit{Uber} belongs to both the top and the bottom.

% AGENTS AND PROFILES
We consider elections in which a finite electorate $N\subset\mathbb{N}$ of \emph{agents} (or \emph{voters}) participate.
%That is, we are working in a model of voting with variable electorates. 
Each agent $i\in N$ casts her vote by reporting a preference ${\pref_i} \in \mathcal{D}$, giving rise to a \emph{profile} of individual preferences. Formally, a profile for the electorate~$N$ is a function from $N$ to $\mathcal{D}$, but for notational convenience we will denote profiles as vectors ${\prof{\pref}} = (\pref_{i})_{i\in N}$. We write $\mathcal{D}^N$ for the set of all such profiles.
We use $\prof{\pref}_{-i}$ to denote the profile where the preference of agent~$i$ in profile $\prof{\pref}$ has been omitted;
and we use $({\pref'}_i,\prof{\pref}_{-i})$ to denote the profile we obtain when we replace the preference of agent~$i$ in $\prof{\pref}$ with ${\pref'}_i$. 
Given two profiles ${\prof{\pref}}\in\mathcal{D}^N$ and ${\prof{\pref}'} \in\mathcal{D}^{N'}$, with $N\cap N'=\emptyset$, we write $(\prof{\pref},\prof{\pref}')\in\mathcal{D}^{N\cup N'}$ for their concatenation, i.e., the profile in which both groups report their preferences.
Given a preference ${\pref}\in\mathcal{D}$, we use $\overline{\prof{\pref}}$ to denote the profile consisting of the $m!$ preferences with the same graph structure as $\pref$, one for each permutation of the set of alternatives.\footnote{As we are going to be interested in rules where the names of the voters are not relevant, this profile is well defined.} 

Finally, a \emph{voting rule} is a function $F : \bigcup_{N\subset\mathbb{N}}\mathcal{D}^N \to 2^A\setminus\{\emptyset\}$, mapping any given profile of preferences (for any finite electorate $N\subset\mathbb{N}$) to a non-empty set of alternatives, the \emph{winners} of the election in question. That is, we are working with voting rules that might return several alternatives that are tied for winning the election.
Next, we introduce some concrete examples for specific definitions of voting rules in this setting.

\subsection{Specific Classes of Voting Rules}\label{sec:votingrules}

For the standard model of voting with linear orders, the class of \emph{positional scoring rules} \citep{young1975,ZwickerHBCOMSOC2016} includes some of the best-known and most widely used rules, such as the Borda rule and the Plurality rule. Following \citet{kruger2020strategic}, we now generalise the definition of this class to our model of voting with partial orders.

A \emph{scoring function} $s : A \times \mathcal{D} \to \mathbb{R}$ maps any given alternative $a\in A$ in a given preference ${\pref} \in \mathcal{D}$ to a real number $s_\pref(a)$, which we call the \emph{score} of~$a$. We are specifically interested in scoring functions that are positional, meaning that the score of an alternative only depends on its position in the given preference. 
Formally, a scoring function~$s$ is \emph{positional} if, for all permutations $\sigma : A \to A$, preferences ${\pref}\in\mathcal{D}$, and alternatives $a\in A$, it is the case that $s_\pref(a) = s_{\sigma(\pref)}(\sigma(a))$, where $\sigma(\pref)=\{\,(\sigma(a),\sigma(b))\mid a\pref b\,\}$. 

Intuitively, a positional scoring function assigns scores to positions in a graph.
Three well-known positional scoring functions for linear orders are the Borda scoring function, giving $m-1$ points to the first alternative, $m-2$ points to the second alternative, and so forth; the Plurality scoring function, assigning one point to the first alternative and zero to the rest; and the Anti-Plurality scoring function, assigning zero points to all the alternatives except for the last-ranked one, which gets $-1$ point.

For every positional scoring function $s$ there is a \emph{positional scoring rule} $F_s$ induced by it, which is the voting rule defined as follows: 
\begin{eqnarray*}
F_s(\prof{\pref}) & = & \argmax_{a\in A}\sum_{i\in N}s_{\pref_i}(a)\quad
\text{for any profile}\ {\prof{\pref}}\in \mathcal{D}^N.
\end{eqnarray*}
Note that applying an affine transformation to a scoring function does not change the induced voting rule.%
\footnote{Applying an affine transformation $x \mapsto \alpha x + \beta$ with $\alpha \in \mathbb{R}_{>0}$ and $\beta \in \mathbb{R}$ to a scoring function~$s$ means replacing, for every preference ${\pref} \in \mathcal{D}$ and every alternative $a\in A$, the original score $s_{\pref}(a)$ with the new score $\alpha\cdot s_{\pref}(a) + \beta$.}
Indeed, while an affine transformation changes the scores assigned to each alternative, it does not affect which alternative receives the highest score and thus wins the election.
Therefore, if $s'$ is an affine transformation of $s$, then $F_s = F_{s'}$, which implies that every positional scoring rule $F_s$ admits infinitely many equivalent representations.

While the idea of positional scoring rules being invariant under affine transformations is familiar from the classical model of voting, in our model, we can furthermore apply what we shall call \emph{additive shifts} to the scores associated with \emph{individual} preferences without affecting outcomes. 
Applying an additive shift of $k$ to a scoring function~$s$ for a subset of the preferences in $\mathcal{D}$ means adding $k$ to the scores of all alternatives under the preferences in the selected set (but not under any of the other preferences). To ensure the scoring function remains positional, the selected set of preferences must be closed under the renaming of alternatives, so all of these preferences must share the same graph structure.\footnote{Observe that in the classical model of voting, where all preferences are linear orders, \emph{all} preferences share the same graph structure. Therefore, in the classical model, the ability to apply additive shifts to individual preferences of the same graph structure is just a special case of the ability to apply affine transformations to the full scoring function. Only in our richer model do these two types of operation actually come truly apart.} 
Since every alternative's score increases by $k$ whenever a voter chooses one of the preferences affected, the aggregate scores shift uniformly and the set of winners remains unchanged. 
The flexibility of being able to apply additive shifts to individual preferences (as long as we respect the constraints imposed by the need to be positional) allows us to normalise scores within a preference to meet specific conditions required for our arguments. This will serve as a key tool in several of our proofs. 

\begin{exa}
Consider the following profile, in which two agents report preferences regarding three alternatives:
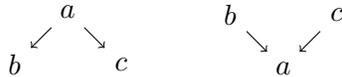
\begin{center}
\begin{tikzpicture} 
\node (a) {$a$}; 
\node (b) [below left of=a] {$b$};
\node (c) [below right of=a] {$c$};
\draw[->] (a) -- (b);
\draw[->] (a) -- (c);
\end{tikzpicture}
$\qquad$
\begin{tikzpicture} 
\node (a) {$a$}; 
\node (b) [above left of=a] {$b$};
\node (c) [above right of=a] {$c$};
\draw[->] (b) -- (a);
\draw[->] (c) -- (a);
\end{tikzpicture}\end{center}
A possible way in which one might generalise the Borda scoring function to the case of partial orders would be to define $s_{\pref}(a)$ as the cardinality of the set $\{\, x \in A \mid a \pref x \,\}$, i.e., by equating the Borda score of an alternative~$a$ with the number of other alternatives dominated by~$a$.
Under this rule, $a$ wins the election with a score of~2, while $b$ and $c$ each receive a score of~1.

One way of generalising the Plurality scoring function to our setting would be to set $s_{\pref}(a)$ to 1 if $\{\, x \in A\mid x\pref a \,\}$ is empty, and to set it to 0 otherwise. Under the corresponding rule, we obtain a three-way tie.

Finally, if we consider, say, the affine transformation $x \mapsto 4 x - 1 $, $a$ gets $6$ points under the Borda-like rule, while $b$ and $c$ get $2$ points each; and under the Plurality-like rule $a$, $b$, and $c$ get $2$ points each.
If we apply an additive shift of $2$ to the first preference only, under the Borda-like rule $a$ gets $4$ points, while $b$ and $c$ get $3$ points each; and under the Plurality-like rule $a$, $b$, and $c$ get $3$ points each.
So the winners indeed do not change.
\end{exa}

\noindent
But the particular way of generalising the Plurality rule suggested in this example is just one of many options. 
We now propose a class of positional scoring rules for partial orders that, we believe, naturally captures the range of options available for generalising the Plurality rule for linear orders.

\begin{defi}\label{def:PluralityClass}
A positional scoring rule $F_s$ induced by a scoring function~$s$ belongs to the \textbf{Plurality Class} if $s$ satisfies the following property: For any given preference ${\pref}\in\mathcal{D}$, there exists a $k'\in\mathbb{R}$ such that for all $a\in T(\pref)$ and $b\in \NT(\pref)$ it is the case that $s_\pref(a)\geq s_\pref(b)=k'$, with this inequality being strict for at least one choice of $a$ and~$b$.   
\end{defi}

\noindent
Thus, we require that all the alternatives not in the top receive the same score $k'$, and that score must not be greater than the scores awarded to the alternatives in the top. But we are free to award different scores to different top alternatives, as long as at least one of the scores is larger than~$k'$. 
This way, we are able to exploit the preference's top structure and, for example, award points to a top alternative according to the number of alternatives it dominates. 
For simplicity, and to highlight the clear analogy to the Plurality rule for linear orders, we are often going to assume that $k'=0$ for any preference ${\pref}\in\mathcal{D}$.\footnote{This can be done w.l.o.g., given that---as noted earlier already---an additive shift can be applied to every preference relation individually.} 

Before we move on, let us briefly comment on one subtlety regarding Definition~\ref{def:PluralityClass}. It says that, for any given scoring function~$s$, the rule $F_s$ induced by it belongs to the Plurality Class if $s$ has certain properties. Thus, a given rule~$F$ belongs to the Plurality class \emph{if we can find} a scoring function~$s$ with those properties such that $F=F_s$. That is, the formulation of Definition~\ref{def:PluralityClass} carefully accounts for the fact that for any given positional scoring rule there are many scoring functions that induce it. According to the definition, a rule belongs to the Plurality Class as soon as \emph{at least one} of those functions has the right properties. But, as we shall see, whenever a rule belongs to the Plurality Class, then in fact \emph{all} of the scoring functions that induce that rule will also have those characteristic properties. This will emerge as a side effect of the characterisation results we are going to prove (and analogous remarks apply to all other classes of positional scoring rules we are about to define).

Now, we stress once more that for two partial orders with different graph structures, the scoring function can assign different scores to the alternatives. The following example illustrates this point.

\begin{exa}
Let $s$ be the scoring function that assigns to any alternative in the top a score that is equal to the number of alternatives it dominates, while to every alternative not in the top it assigns the score of~$0$. 
Here are the scores assigned to the alternatives for two possible preferences:

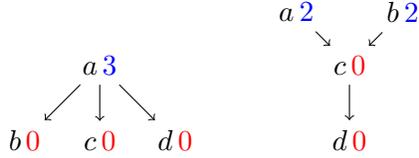
\begin{center}
\begin{tikzpicture} 
\node (1) {$a$\,\color{blue}{$3$}};
\node (2) [below right of=1] {$c$\,\color{red}{$0$}};
\node (3) [below left of=1] {$b$\,\color{red}{$0$}};
\node (4) [above right of=2] {$e$\,\color{blue}{$1$}};
\node (5) [below of=3] {$d$\,\color{red}{$0$}};
\draw[->] (4) -- (2);
\draw[->] (1) -- (3);
\draw[->] (1) -- (2);
\draw[->] (3) -- (5);
\end{tikzpicture}
\qquad
\begin{tikzpicture}
\node (1) {$a$\,\color{blue}{$2$}};
\node (2) [below right of=1] {$c$\,\color{red}{$0$}};
\node (3) [above right of=2] {$b$\,\color{blue}{$2$}};
\node (4) [below right of=2] {$d$\,\color{red}{$0$}};
\node (5) [right of=3] {$e$\,\color{blue}{$1$}};
\draw[->] (1) -- (2);
\draw[->] (3) -- (2);
\draw[->] (2) -- (4);
\draw[->] (5) -- (4);
\end{tikzpicture} 
\end{center}
The positional scoring rule~$F_s$ induced by this scoring function~$s$ belongs to the Plurality Class.
\end{exa}

\noindent
This example can be viewed as a hybrid of the Borda and Plurality scoring rules: the Borda rule, where alternatives receive points based on the number of dominated options, and the Plurality rule, where a voter contributes only to her most preferred alternatives. 
But we can also take the stance that all alternatives that are part of the top set in a voter's preference should receive the same score. Thus, we consider the following class.

\begin{defi}
A positional scoring rule $F_s$ in the Plurality Class induced by a scoring function~$s$ belongs to the \textbf{Simple Plurality Class} if $s$ satisfies the following property: For any given preference ${\pref}\in\mathcal{D}$, there exists a $k\in\mathbb{R}$ such that $s_\pref(a)=k$ for all $a\in T(\pref)$.   
\end{defi}

\noindent
Thus, $F_s$ belongs to the Simple Plurality Class if, for any given preference ${\pref}\in\mathcal{D}$, there exist $k,k'\in\mathbb{R}$ such that $k=s_\pref(a) > s_\pref(b) = k'$ for all $a\in T(\pref)$ and $b\in \NT(\pref)$. 

Observe that for rules in this class, we still can make $k$, the score awarded to each alternative in a given top set, dependent on the preference~$\pref$ at hand. 
For example, we could consider a rule that awards $2$~points to alternatives alone in the top; $2$~points to alternatives in the top together with all but one of the other alternatives; $1$~point to alternatives in the top in all other cases; and $0$~points to alternatives not in the top. This rule rewards and punishes alternatives in extreme positions: when only one alternative is in the top, this can be deemed without any doubt the best alternative, so it deserves to obtain $2$~more points that the rest; while when there is only one alternative not in the top, that one alternative can without any doubt be deemed the worst one, so it deserves to obtain $2$~points less than the rest.

While the use of such a rule can be defended in some situations, in others we might want to insist that the score an alternative earns for being placed in the top should go up (relative to the score those outside of the top receive) as the number of alternatives it shares that top spot with goes down. This suggests yet another narrowing of our class of voting rules.

\begin{defi}
A positional scoring rule $F_s$ in the Simple Plurality Class induced by a scoring function~$s$ belongs to the \textbf{Monotonic Simple Plurality Class} if $s$ satisfies the following property: For any two preferences ${\pref},{\pref'}\in\mathcal{D}$ and any four alternatives $a\in T(\pref)$, $b\in \NT(\pref)$, $a'\in T(\pref')$, and $b'\in \NT(\pref')$ it is the case that $|T(\pref)| \leq |T(\pref')|$ entails $s_\pref(a) - s_\pref(b) \geq s_{\pref'}(a')-s_{\pref'}(b')$.
\end{defi}

\noindent
Restricting the scoring function even further, let us now focus on the positional scoring rules in the Plurality Class that assign the same score to all the top alternatives.  
Such a Plurality rule is \emph{uniform} in the sense that the score assigned to a top alternative does not depend on the structure of the preference but is the same across all possible preferences, and the same is true for the alternatives that are not in the top.
Note that, in view of what we said earlier about positional scoring rules being invariant under affine transformations, there is in fact only a single rule in the Plurality Class that is uniform---but it can be represented by infinitely many different scoring functions. We call this unique rule the \emph{Uniform Plurality} rule.
The canonical way of defining it is as the rule that assigns 1~point to every alternative in the top and 0~points to all other alternatives. 

\begin{defi}
The \textbf{Uniform Plurality} rule $F_s$ is the rule belonging to the Monotonic Simple Plurality Class that is induced by the scoring function $s$ with $s_{\pref}(a) = 1$ and $s_{\pref}(b) = 0$ for all preferences ${\pref} \in \mathcal{D}$, alternatives $a \in T(\pref)$, and alternatives $b \in \NT(\pref)$.
\end{defi}

\noindent
The Uniform Plurality rule might be regarded as the most natural extension of the Plurality rule for linear orders to our model of voting with partial orders, and it is clearly appealing due to its simplicity.
It captures the most basic intuition about the Plurality rule for linear orders, selecting the alternatives that most frequently appear in the top positions of the ballots. Having said this, focusing \emph{only} on the Uniform Plurality rule would not do justice to the rich variety in which voters can express themselves in our model of voting with partial orders; so certain other rules belonging to the Plurality Class also have their place.

We now move on to defining corresponding classes of positional scoring rules for partial orders that generalise the Anti-Plurality rule for linear orders. 

\begin{defi}
A positional scoring rule $F_s$ induced by a scoring function~$s$ belongs to the \textbf{Anti-Plurality Class} if $s$ satisfies the following property: For any given preference ${\pref}\in\mathcal{D}$, there exists a $k\in\mathbb{R}$ such that for all $a\in \NB(\pref)$ and $b\in B(\pref)$ it is the case that $k=s_\pref(a)\geq s_\pref(b)$, with this inequality being strict for at least one choice of $a$ and~$b$.
\end{defi}

\noindent 
Again, for simplicity, and to highlight the clear analogy to the Anti-Plurality rule for linear orders, we are often going to assume that $k=0$.

\begin{defi}
A positional scoring rule $F_s$ in the Anti-Plurality Class induced by a scoring function~$s$ belongs to the \textbf{Simple Anti-Plurality Class} if $s$ satisfies the following property: For any given preference ${\pref}\in\mathcal{D}$, there exists a $k'\in\mathbb{R}$ such that $s_\pref(b)=k'$ for all $b\in B(\pref)$.
\end{defi}

\begin{defi}
A positional scoring rule $F_s$ in the Simple Anti-Plurality Class induced by a scoring function~$s$ belongs to the \textbf{Monotonic Simple Anti-Plurality Class} if $s$ satisfies the following property: For any two preferences ${\pref},{\pref'}\in\mathcal{D}$ and any four alternatives $b\in B(\pref)$, $c\in \NB(\pref)$, $b'\in B(\pref')$, and $c'\in \NB(\pref')$, it is the case that  $|B(\pref)| \leq |B(\pref')|$ entails $s_\pref(c) - s_\pref(b) \geq s_{\pref'}(c') - s_{\pref'}(b')$.
\end{defi}

\noindent
Thus, the alternatives that are part of the smaller bottom set receive the lower scores, while for the Monotonic Simple Plurality Class the alternatives in the smaller top set receive the higher scores. 
We again identify a single most canonical rule for this class.

\begin{defi}
The \textbf{Uniform Anti-Plurality} rule $F_s$ is the rule belonging to the Monotonic Simple Anti-Plurality Class that is induced by the scoring function $s$ with $s_{\pref}(a) = 0$ and $s_{\pref}(b) = -1$ for all preferences ${\pref} \in \mathcal{D}$, alternatives $a \in \NB(\pref)$, and alternatives $b \in B(\pref)$.
\end{defi}

\noindent
Again, this rule captures the most basic intuition about the Anti-Plurality rule for linear orders, selecting those alternatives that least often appear in the bottom of a preference.

\subsection{Axioms}\label{sec:axioms}

Next, we introduce a number of axioms that each encode a normative requirement that, one might argue, any reasonable voting rule should satisfy.
The first four axioms, Anonymity, Neutrality, Reinforcement, and Continuity, are very common in the literature on scoring rules and are used by \citet{young1975} to characterise them.
The first two are symmetry requirements that essentially state that the names of the agents and alternatives should not be considered important for choosing a winner. 
We formulate Anonymity in terms of a bijection $\mu:N\rightarrow N'$ for any two electorates $N,N'\subset\mathbb{N}$ with the same number of agents. Similarly, we formulate Neutrality in terms of a permutation $\sigma:A\rightarrow A$ on the set of alternatives, which extends to both sets $S\subseteq A$ of alternatives and preferences ${\pref} \in \mathcal{D}$ in the natural manner:
$\sigma(S) = \{\,\sigma(x) \mid x\in S\,\}$ and $\sigma(\pref)=\{\,(\sigma(x),\sigma(y))\mid x\pref y\,\}$.

\begin{axiom}[Anonymity]
For any two electorates $N,N'\subset\mathbb{N}$ with the same number of agents and any bijection $\mu:N\rightarrow N'$, it should be the case that \mbox{$F((\pref_i)_{i\in N})$} = \mbox{$F((\pref_{\mu(i)})_{i\in N})$}.  
\end{axiom}

\begin{axiom}[Neutrality] 
For any profile ${\prof{\pref}} = (\pref_i)_{i\in N}$ and any permutation $\sigma:A\rightarrow A$, it should be the case that $\sigma(F((\pref_i)_{i\in N})) = \mbox{$F((\sigma(\pref_i))_{i\in N})$}$.
\end{axiom}

\noindent
Thus, Anonymity says that changing the names of the voters should not change the winners, while Neutrality says that any change in the names of the alternatives in the profile should then also be applied to the original winners to obtain the new winners.

Reinforcement, which is also known as Consistency, states that if two disjoint electorates agree on some winning alternatives, then those alternatives must be selected when we consider the profile where the members of both electorates cast their votes.
 
\begin{axiom}[Reinforcement] 
For any two profiles ${\prof{\pref}}\in\mathcal{D}^N$ and ${\prof{\pref}'}\in\mathcal{D}^{N'}$ with $N\cap N'=\emptyset$ such that
$F(\prof{\pref})\cap F(\prof{\pref}')\neq\emptyset$, it should be the case that $F(\prof{\pref},\prof{\pref}')=F(\prof{\pref})\cap F(\prof{\pref}')$.
\end{axiom}

\noindent
Continuity, alternatively called the Archimedean property by \citet{smith1973aggreg} and \citet{young1975} and Overwhelming Majority by \citet{myerson1995axiomatic}, stipulates that a small number of agents cannot completely override the decision of a large majority (although they might still break ties in favour of some of the alternatives selected by the majority). 

\begin{axiom}[Continuity] 
For any two profiles ${\prof{\pref}}\in\mathcal{D}^N$ and ${\prof{\pref}'}\in\mathcal{D}^{N'}$ such that $N\cap N'=\emptyset$, there should exist a bound $K$ such that, for every natural number $k>K$, the following inclusion holds: 
\begin{eqnarray*}
F(\underbrace{\prof{\pref},\cdots,\prof{\pref}}_{k\ \text{times}},\prof{\pref}')& \subseteq & F(\prof{\pref}),
\end{eqnarray*} 
where, to assign identities to the voters of the duplicated profile, we use the smallest numbers in the set $\mathbb{N}\setminus\{N\cup N'\}$ of numbers not yet used. 
\end{axiom}

\noindent
In the sequel, we are going to write $k{\prof{\pref}}$ as a shorthand for the profile consisting of a concatenation of $k$ copies of the profile~$\prof{\pref}$.

The next two axioms are minimal responsiveness and efficiency requirements for a voting rule.
The first one is a version of the Faithfulness axiom due to \citet{YOUNG197443}, and the second one is a version of Averseness introduced by \citet{KURIHARA2018110}.
The former states that, for the special case of single-voter profiles, a voting rule must select alternatives that are in the top of that agent.
Thus, dominated alternatives are not selected.
The latter states that for such profiles, a rule should not select \emph{all} of the bottom alternatives. So, not all dominated alternatives are selected.

\begin{axiom}[Partial Faithfulness]
For any single-voter profile ${\pref}\in\mathcal{D}$, it should be the case that \mbox{$F(\pref)\subseteq T(\pref)$}.  
\end{axiom}

\begin{axiom}[Partial Averseness]
For any single-voter profile ${\pref}\in\mathcal{D}$ with $B(\pref)\neq A$, it should be the case that \mbox{$B(\pref)\nsubseteq F(\pref)$}.  
\end{axiom} 

\noindent
Note that Partial Averseness is a weaker requirement than Partial Faithfulness, in the formal sense of being implied by the latter.
Now we present stronger versions of these axioms.
\begin{axiom}[Faithfulness]
For any single-voter profile ${\pref}\in\mathcal{D}$, it should be the case that \mbox{$F(\pref)= T(\pref)$}.  
\end{axiom}

\begin{axiom}[Averseness]
For any single-voter profile ${\pref}\in\mathcal{D}$ with $B(\pref)\neq A$, it should be the case that \mbox{$B(\pref)\cap F(\pref)=\emptyset$}.  
\end{axiom} 

\noindent
It is easy to see that in the case where preferences are linear orders, the partial and non-partial versions of these axioms coincide.

Our remaining axioms are all inspired by the approval voting framework introduced by \citet{bramsfishburnapporvalvoting1978}. 
We adapt axioms from \citet{ALCALDEUNZU20091187} and \citet{BardalEndrissGEB2025}, and introduce new ones tailored to situations where one might want to use an Anti-Plurality rule. 
The next two axioms, inspired by the Congruity axiom of \citeauthor{ALCALDEUNZU20091187}, can be interpreted as agreement properties: 
if we add to a given electorate additional voters who, in an immediate sense, agree with the collective choice of that electorate, then the new collective choice should not be altered in any significant way as far as this agreement is concerned.
More specifically, T-Congruity states that adding to an electorate that does not elect~$x$ additional voters who do not rank $x$ in the top should result in a new electorate that still does not elect~$x$.
Similarly,  B-Congruity states that adding to an electorate that elects~$x$ additional voters who do not rank $x$ at the bottom of their preferences should not result in $x$ now losing.

\begin{axiom}[T-Congruity]
Let $x\in A$, ${\prof{\pref}}\in\mathcal{D}^N$, and ${\prof{\pref}'}\in\mathcal{D}^{N'}$ be such that $N\cap N'=\emptyset$, $x\notin F(\prof{\pref})$, and $x\notin T(\pref'_i)$ for all $i\in N'$.
Then $x\notin F(\prof{\pref},\prof{\pref}')$.    
\end{axiom}

\begin{axiom}[B-Congruity]
Let $x\in A$, ${\prof{\pref}}\in\mathcal{D}^N$, and ${\prof{\pref}'}\in\mathcal{D}^{N'}$ be such that $N\cap N'=\emptyset$, $x\in F(\prof{\pref})$, and $x\notin B(\pref'_i)$ for all $i\in N'$.
Then $x\in F(\prof{\pref},\prof{\pref}')$.     
\end{axiom}

\noindent
Our last four axioms capture the idea that elected alternatives should not be harmed when a voter withdraws support from other alternatives. We distinguish between two cases according to the position, top or bottom, of the alternatives being removed. 
Specifically, if an individual weakens her support for a non-elected alternative, either by dropping it from her top or by adding it to her bottom, then any already-elected alternative she supports must remain elected.
Following the naming conventions in the work of \citet{ALCALDEUNZU20091187} and \citet{BardalEndrissGEB2025}, we call the first of these axioms Contraction. It imposes constraints on situations where we contract the top set. In analogy to this choice of name, we call the other type of axiom Expansion, given that it imposes constraints when we expand the bottom set. For both axioms, we formulate both a weak and a strong variant.

\begin{axiom}[Contraction]
Let ${\prof{\pref}}\in\mathcal{D}^N$ and ${\pref'_i}\in\mathcal{D}$ be such that \mbox{$T(\pref'_i)\subseteq T(\pref_i)$}.
Then \mbox{$F(\prof{\pref})\cap T(\pref'_i)\subseteq F(\pref'_i,\prof{\pref}_{-i})$}.  
\end{axiom}

\begin{axiom}[Strong Contraction]
Let ${\prof{\pref}}\in\mathcal{D}^N$ and ${\pref'_i}\in\mathcal{D}$ be such that \mbox{$T(\pref'_i)\subseteq T(\pref_i)$}. 
Then \mbox{$F(\prof{\pref})\setminus (T(\pref_i)\setminus T(\pref_i')) \subseteq F(\pref'_i,\prof{\pref}_{-i})$}.
\end{axiom}

\begin{axiom}[Expansion]
Let ${\prof{\pref}}\in\mathcal{D}^N$ and ${\pref'_i}\in\mathcal{D}$ be such that \mbox{$B(\pref_i)\subseteq B(\pref'_i)$}.
Then \mbox{$F(\prof{\pref})\cap B(\pref_i)\subseteq F(\pref'_i,\prof{\pref}_{-i})$}.    
\end{axiom}

\begin{axiom}[Strong Expansion]
Let ${\prof{\pref}}\in\mathcal{D}^N$ and ${\pref'_i}\in\mathcal{D}$ be such that \mbox{$B(\pref_i)\subseteq B(\pref'_i)$}. 
Then \mbox{$F(\prof{\pref})\setminus (B(\pref'_i)\setminus B(\pref_i)) \subseteq F(\pref'_i,\prof{\pref}_{-i})$}.
\end{axiom}

\noindent
We can see that each of the basic two axioms is implied by its strong counterpart by considering the identity $C\setminus(B\setminus A)=(C\cap A)\cup (C\setminus B)$.
This identity also points to the conceptual difference between versions: the strong ones state that the only alternatives that might be harmed, if any, are the ones that got removed from the top (for Strong Contraction) or added to the bottom (for Strong Expansion), while the weaker ones state that any alternative that is not in the contracted top (for Contraction) or within the expanded bottom (for Expansion) might be harmed.

%%%%%%%%%%%%%%%%%%%%%%%%%%%%%%%%%%%%%%%%%%%%%%%%%%%%%%%%%%%%%%%%%%%%%%%%%%%%%%%%
\section{Results}\label{sec:results}
%%%%%%%%%%%%%%%%%%%%%%%%%%%%%%%%%%%%%%%%%%%%%%%%%%%%%%%%%%%%%%%%%%%%%%%%%%%%%%%%

In this section, we provide axiomatic characterisations of the classes of voting rules defined in Section~\ref{sec:votingrules} in terms of the axioms put forward in Section~\ref{sec:axioms}.

For the class of positional scoring rules, it is well-understood that they can be characterised in terms of Anonymity, Neutrality, Reinforcement, and Continuity, and results of this kind have been obtained for a number of different models. 
Originally, \citet{smith1973aggreg} and \citet{young1975} proved this to be the case for the classical model of voting with linear orders. 
Our model is closest to that of \cite{kruger2020strategic}, who obtained the same kind of characterisation for the case of voting with incomplete preferences that are merely assumed to be acyclic (but not necessarily transitive).
Although their result does not directly imply the corresponding result for our domain (which is a subdomain of the domain considered by Kruger and Terzopoulou),\footnote{It might, at first, seem counterintuitive that a characterisation result established for one domain does not immediately transfer to its subdomains, but there are numerous counterexamples for such transfers in the literature. 
One example is the Gibbard-Satterthwaite Theorem \citep{gibbard1973manipulation,satterthwaite1975strategy}. On the full domain of linear preferences, it establishes a characterisation of the family of dictatorships in terms of strategy-proofness and surjectivity. 
But on the single-peaked domain, these axioms admit also many other rules, such as Black's median-voter rule \citep{black1948rationale}.} a careful examination of their proof shows that it holds for our model as well (so we provide here only a sketch of the proof).
We note that \citet{myerson1995axiomatic} and \citet{PIVATO2013210} proved similar results for a model without an ordering assumption, and in principle it would also be possible to derive our Theorem~\ref{thm:psr} as a corollary to their results. 
While Theorem~\ref{thm:psr} is technically nontrivial, it is ultimately not unexpected, given that analogous characterisations of scoring rules have been shown to hold across a variety of domains, as illustrated by the aforementioned works.

\begin{teo}\label{thm:psr}
A voting rule satisfies Anonymity, Neutrality, Reinforcement, and Continuity if, and only if, it is a positional scoring rule.
\end{teo}

\begin{proof}[Proof sketch]
The ``if'' direction is routine. For the ``only if'' direction, we follow the proof of \citet{kruger2020strategic}, adapted to our domain of partial orders. 
By \citet{myerson1995axiomatic}, any rule satisfying Anonymity, Neutrality, Reinforcement, and Continuity is a scoring rule $F_s$ for some scoring function $s$. It remains to show that $s$ can be chosen to be positional. Following \citet{kruger2020strategic}, it suffices to verify that for all alternatives $x, y \in A$ and preferences ${\pref} \in \mathcal{D}$ we get:
\begin{align*}
s_{\pref}(x) - s_{\pref}(y) &= s_{\pref^{xy}}(y) - s_{\pref^{xy}}(x), \\
s_{\pref}(z) - s_{\pref}(y) &= s_{\pref^{xy}}(z) - s_{\pref^{xy}}(x) \quad \text{for all } z \neq x, y.
\end{align*}
The profiles constructed in the verification of these conditions consist entirely of permuted copies of preferences in $\mathcal{D}$. Since applying a permutation~$\sigma$ to a strict partial order preserves irreflexivity, antisymmetry, and transitivity, all constructed profiles remain within $\mathcal{D}$. The remainder of the argument given by \citet{kruger2020strategic} therefore applies without modification.
\end{proof}

\noindent
This result is very useful, as it allows us to directly focus on voting rules that are positional scoring rules, which have a very well-defined structure.

\subsection{The Plurality Class}\label{sec:pluraclass}

We are now ready to provide axiomatic characterisations of the Plurality Class and its subclasses we defined earlier.

\begin{teo}\label{theo:plurality}
A voting rule satisfies Anonymity, Neutrality, Reinforcement, Continuity, Partial Faithfulness, and T-Congruity if, and only if, it belongs to the Plurality Class. 
\end{teo}  
 
\begin{proof}
Let $F$ be a voting rule that satisfies the six axioms.
First, we show that for the particular case where the profile of preferences consists of only linear orders, $F$ must be the standard Plurality rule.
By Theorem~\ref{thm:psr}, we know that $F$ is a positional scoring rule.
Given a linear order~$L$, let $s_k$ be the score assigned to the alternative in position $k$.
We need to show that $s_1>s_k$  and $s_k=s_{k'}$ for all $k,k'\neq 1$.
By Partial Faithfulness, it is the case that $s_1>s_k$ for all $k\neq 1$.
Now, for the sake of contradiction, suppose there are two positions $k,k'\neq 1$ such that $s_{k}\neq s_{k'}$.
W.l.o.g., we can assume that $s_{k}$ and $s_{k'}$ are the second and third highest scores, respectively, and that $s_{k'}=0$.
     
Let $L_1$ be a linear order with $a$ in the first position, $b$ in position $k$, and $c$ in position $k'$; and let $L_2$ be a linear order with $b$ in the first position, $a$ in position $k$, and $c$ in position $k'$.
Now consider a profile ${\prof{\pref}}=(\ell L_1,\ell 'L_2)$ with
\begin{eqnarray}\label{lconstrain}
    \ell & = & \ell'+1,
\end{eqnarray}
with $\ell$ and $\ell'$ being sufficiently large such that:
\begin{eqnarray}\label{scoreconstrain}
    \left\lceil \dfrac{s_1-s_{k}}{s_{k}} \right\rceil s_1 &\leq & \ell s_{k}+\ell's_1+ \left\lceil \dfrac{s_1-s_{k}}{s_{k}} \right\rceil s_{k}.
\end{eqnarray}
We have $F(\prof{\pref})=\{a\}$ and the aggregate score for $a$ is $\ell s_1+\ell's_{k}$, the one for $b$ is $\ell s_{k}+\ell's_1$, the one for $c$ is $0$, and for any other alternative it is at most $(\ell+\ell')s_{k}$. 
Now let $L_3$ be a linear order with $c$ in the first position, $b$ in position $k$, and $a$ in position $k'$.
Consider the profile ${\prof{\pref}'}=(\lceil \frac{s_1-s_{k}}{s_{k}} \rceil L_3)$.
The aggregated scores in $F(\prof{\pref},\prof{\pref}')$ are:
$\ell s_1+\ell's_{k}$ for $a$; 
$\ell s_{k}+\ell's_1+\lceil \frac{s_1-s_{k}}{s_{k}} \rceil s_{k}$ for $b$; 
$\lceil \frac{s_1-s_{k}}{s_{k}} \rceil s_1$ for $c$; 
and for the other alternatives at most $(\ell+\ell'+\lceil \frac{s_1-s_{k}}{s_{k}} \rceil)s_{k}$. 

Given Equations~(\ref{lconstrain}) and~(\ref{scoreconstrain}), we have that $b\in F(\prof{\pref},\prof{\pref}')$, violating T-Congruity.
So we can conclude that $s_{k}=s_{k'}=0$ for all $k,k'\neq 1$, and thus that $F$ acts as the standard Plurality rule when preferences are linear.\footnote{To the best of our knowledge, this constitutes a new characterisation of the Plurality rule for linear orders that is independent of any of the previously known characterisations \citep{ching1996simple,Merlin1999ImplementationOS,sekiguchi2012characterization}.}

Now we show that for any given preference ${\pref}\in\mathcal{D}$ the scores awarded by $F$ are those of a voting rule in the Plurality Class.
To do so, we need to show that there exists a $k\in\mathbb{R}$ such that, for all $c\in \NT(\pref)$, we get $s_\pref(c)=k$; that there is an $a\in T(\pref)$ such that $s_\pref(a)> k$; and that for all $b\in T(\pref)$ we get $s_\pref(b)\geq k$. 
W.l.o.g., we can assume that there is a $c\in \NT(\pref)$ in a position such that $s_\pref(c)=k=0\geq s_\pref(d)$ for all $d\in \NT(\pref)$ (so $c$ gets the highest possible score for a non-top alternative).
By Partial Faithfulness, there is an alternative $a\in T(\pref)$ in a position such that $s_\pref(a)>0$.
Let $a\in F(\pref)$ (so $a$ gets the highest score possible).
Now, suppose there is an alternative $b\in T(\pref)$ such that $s_\pref(b)<0$.
Let $L_1$ and $L_2$ be two linear orders such that $b$ and $c$ are in the first position, respectively.
Now consider the profile ${\prof{\pref}}=(\ell L_1,\ell'L_2)$ such that:
\begin{eqnarray}\label{lconstrain2}
    \ell & = & \ell'+1
\end{eqnarray}
and
\begin{eqnarray}\label{scoreconstrain2}
    \left\lceil \dfrac{1}{-s_\pref(b)} \right\rceil s_\pref(a) & \leq & \ell'.
\end{eqnarray}
Let ${\prof{\pref}}=(\ell L_1,\ell'L_2)$, so we have $F(\prof{\pref})=\{b\}$.
Now consider the profile ${\prof{\pref}'}=(\prof{\pref},\lceil \frac{1}{-s_\pref(b)} \rceil\pref)$.
The aggregated score for $c$ is $\ell'$; the one for $b$ is $\ell+\lceil \frac{1}{-s_\pref(b)} \rceil s_\pref(b)$; and any other alternative gets at most $\lceil \frac{1}{-s_\pref(b)} \rceil s_\pref(a)$.
Then, by Equations~(\ref{lconstrain2}) and~(\ref{scoreconstrain2}), $c\in F(\prof{\pref}')$ and thus T-Congruity is violated.
We can conclude that when $b\in T(\pref)$, it is the case that $s_\pref(b)\geq 0$.

Using a similar argument, we can show that there cannot be a position for a non-top alternative where it gets a negative score.
Thus, we can conclude that $s_\pref(c)=0$ for all $c\in \NT(\pref)$. So $F$ indeed belongs to the Plurality Class.
\end{proof}

\noindent
If $F$ is a positional scoring rule, Faithfulness implies that all the alternatives in the top of a preference receive the same score.

\begin{coro}\label{coro:simpleplurality}
A voting rule satisfies Anonymity, Neutrality, Reinforcement, Continuity, Faithfulness, and T-Congruity if, and only if, it belongs to the Simple Plurality Class.    
\end{coro}

\noindent
Contraction requires the scoring rule to assign more points to top alternatives when they are less accompanied by other top alternatives. 
This axiom plays a crucial role in our characterisation of the Monotonic Simple Plurality Class.\footnote{Theorem~\ref{theo:monosimpleplurality} continues to hold if we relax Faithfulness to Partial Faithfulness, albeit at the expense of a somewhat more complex (but still very similar) proof, the details of which we omit in the interest of brevity. This point will become relevant in Section~\ref{sec:appvot}, for Corollary~\ref{coro:sizeapppfcontr}.}

\begin{teo}\label{theo:monosimpleplurality}
A voting rule satisfies Anonymity, Neutrality, Reinforcement, Continuity, Faithfulness, and Contraction if, and only if, it belongs to the Monotonic Simple Plurality Class.
\end{teo}

\begin{proof}
Let ${\pref}\in\mathcal{D}$ and assume w.l.o.g.\ that $T(\pref)\neq A$. 
By Theorem~\ref{thm:psr}, we know that $F$ is a positional scoring rule.
By Faithfulness, it is the case that $F(\pref)=T(\pref)$.
Thus, we can conclude that there exists a value $k$ such that $s_\pref(a)=k$ for all $a\in T(\pref)$, and $k>s_\pref(c)$ for all $c\in \NT(\pref)$.

Now assume that $|\NT(\pref)|\geq 2$ and that there are two positions where alternatives $c,d\in \NT(\pref)$ get scores $s_\pref(c)>s_\pref(d)$.
Consider the profile $\overline{\prof{\pref}}$.\footnote{We remind the reader that $\overline{\prof{\pref}}$ is the profile consisting of the $m!$ individual preferences that have the same graph structure as~$\pref$.}
By Neutrality and Anonymity, $F(\overline{\prof{\pref}})=A$.
Now consider $\pref_i^{cd}$, that is, the preference $\pref$ with the labels of $c$ and $d$ switched and submitted by individual~$i$.
We obtain that $F(\pref_i^{cd},\overline{\prof{\pref}}_{-i})=\{d\}$, as $d$ gets a higher score.
But, as $T(\pref)=T(\pref_i^{cd})\cap F(\overline{\prof{\pref}})$, by Contraction we have that \mbox{$T(\pref)\subseteq (\pref_i^{cd},\overline{\prof{\pref}}_{-i})$}, obtaining a contradiction.
Thus, there is a $k$ such that $s_\pref(c)=k$ for all $c\in \NT(\pref)$.

As Contraction only considers the alternatives in the top of a preference and not the graph structure of the top, and by Neutrality, we have that for any two preferences $\pref$ and $\pref'$ such that $|T(\pref)|=|T(\pref')|$, it is the case that $s_\pref(a)=s_{\pref'}(b)$ for any $a\in T(\pref)$ and $b\in T(\pref')$, or any $a\in \NT(\pref)$ and $b\in \NT(\pref')$.

We can normalise the scores by means of additive shifts, such that for all ${\pref}\in\mathcal{D}$ we get \mbox{$s_\pref(c)=0$ for all $c\in \NT(\pref)$}.
Now consider two preferences ${\pref},{\pref'}\in\mathcal{D}$ where \mbox{$|T(\pref)|<|T(\pref')|$}, and $k$ and $k'$ are the scores for the top alternatives of $\pref$ and $\pref'$, respectively. 
W.l.o.g., we can assume that \mbox{$T(\pref)\subset T(\pref')$}.
Suppose $k<k'$.
Consider the profile $\overline{\prof{\pref}'}$.
By Neutrality and Anonymity, \mbox{$F(\overline{\prof{\pref}'})=A$}.
We have that \mbox{$T(\pref)\cap F(\overline{\prof{\pref}'})=T(\pref)$} and thus, assuming $i$ has preference $\pref$, by Contraction, $T(\pref)\subseteq F(\pref,\overline{\prof{\pref}'}_{-i})$, the profile where $\pref'$ is replaced by $\pref$.
But in $F(\pref,\overline{\prof{\pref}'}_{-i})$, all the alternatives in $T(\pref')$ get their scores reduced, while those not in the top of $T(\pref')$ remain with the same score, obtaining $F(\pref,\overline{\prof{\pref}'}_{-i})=A\setminus{T(\pref')}$.
With this contradiction, we can conclude that $k'\leq k$, and thus, that $F$ belongs to the Monotonic Simple Plurality Class.
\end{proof}

\noindent
We now introduce a further condition known as Tops-Only. 
Importantly, we think of Tops-Only as a purely technical property, with no immediate normative appeal in and of its own. 
However, as we are going to see soon, it will turn out to be implied by another property that does have such appeal, namely Strong Contraction.
Tops-Only requires that only the top-ranked alternative(s) in each voter’s preference matter to the outcome.
Formally, a voting rule satisfies Tops-Only if, for any two profiles ${\prof{\pref}},{\prof{\pref}'}\in\mathcal{D}^N$ with $T(\pref_i)=T(\pref_i')$ for all $i\in N$, it is the case that $F(\prof{\pref})=F(\prof{\pref}')$.\footnote{For further motivation regarding this property, and its use for characterisation results, we refer to \citet{sekiguchi2012characterization} and \citet{kelly2016characterizing}.
Earlier works that have considered this property include those by \citet{moulin1980strategy}, \citet{barbera1991voting}, and \citet{korayselfselective2000}.}
We show that any rule satisfying Strong Contraction necessarily depends only on the voters’ top-ranked alternatives.

\begin{lema}\label{lema:coherencetopsonly}
Strong Contraction implies Tops-Only.    
\end{lema}

\begin{proof}
Let ${\prof{\pref}},{\prof{\pref}'}\in\mathcal{D}^N$ be two profiles such that, for all $i\in N$, it is the case that $T(\pref_i)=T(\pref'_i)$.
For any $i\in N$, we have that $T(\pref_i)\setminus T(\pref'_i)=\emptyset$, and thus, by Strong Contraction, $F(\prof{\pref})\subseteq F(\pref'_i,\prof{\pref}_{-i})$.
By iterating this process for all $i\in N$, we obtain that $F(\prof{\pref})\subseteq F(\prof{\pref}')$.
With the same reasoning, we obtain that $F(\prof{\pref}')\subseteq F(\prof{\pref})$, concluding that $F(\prof{\pref})= F(\prof{\pref}')$.
\end{proof}

\noindent
To characterise the Uniform Plurality rule, we build on two different results.
First, we use the result by \citet{sekiguchi2012characterization} who characterises the Plurality rule for linear orders in terms of Anonymity, Neutrality, Reinforcement, Partial Faithfulness, and Tops-Only.
And then we use the result by \citet{PIVATO2013210}, who shows that if a rule satisfies Anonymity, Neutrality and Reinforcement, it is a \emph{composite scoring rule}.
This is a rule that repeatedly applies scoring functions in order to break possible ties.

\begin{teo}\label{theo:uniformplurality}
A voting rule satisfies Anonymity, Neutrality, Reinforcement, Partial Faithfulness, and Strong Contraction if, and only if, it is the Uniform Plurality rule. 
\end{teo}

\begin{proof}
By Lemma~\ref{lema:coherencetopsonly}, we have that any rule satisfying our axioms also satisfies Tops-Only.
By a result due to \citet{sekiguchi2012characterization}, we know that when the profile consists of linear orders, a voting rule satisfies all the axioms if, and only if, it is the Plurality rule.
By applying an additive shift for the linear orders, we can normalise the rule such that the score for the non-top alternatives of a linear order is equal to $0$, and let $s_L>0$ be the score for the top alternative.

By the aforementioned result due to \citet{PIVATO2013210}, we know that if a rule satisfies Anonymity, Neutrality and Reinforcement, it is a composite scoring rule.
Thus, our rule involves a sequence of scoring functions $s^1,\ldots,s^k$ where $s^{k'}$ is used to break ties not yet broken by $s_1,\ldots,s_{k'-1}$.

Let ${\pref}\in\mathcal{D}$ be a generic preference.
By Partial Faithfulness and Tops-Only, we have that $F(\pref)=T(\pref)$ (to see this, consider the top with all symmetric alternatives).
This implies that for all $i\in\{1,\ldots,k\}$, all \mbox{$a\in T(\pref)$}, and all $b\in \NT(\pref)$, $s_\pref^i(a)=k_i$ and $k_i\geq s_\pref^i(b)$, with at least one $t$ such that $k_t> s_\pref^t(b)$.

Now we assume there are an $s_\pref^j$ and two positions not at the top such that alternatives $c,d\in \NT(\pref)$ get scores $x=s_\pref^j(c)>s_\pref^j(d)=y$.
Let $\prof{L}=(\pref_1,\pref_2)$ with $\pref_1$ and $\pref_2$ being linear orders such that $c\pref_1 d\pref_1\ldots$ and  $d\pref_2 c\pref_2\ldots$, and all the alternatives different from $c$ and $d$ being ordered the same.
As $F$ is the Plurality rule when considering linear orders, we have that $F(\prof{L})=\{c,d\}$.
As for linear orders $F$ acts as the plurality rule, we obtain that the scores are $s_L$ for $c$ and $d$, and 0 for the other alternatives.
Let $n$ be such that $ns_L+x>\max_i k_i$.
Thus, $F(n\prof{L},\pref)=\{c\}$, as $c$ gets a higher score than $d$, and any other alternative gets at most $\max_i k_i$.
Now consider $\pref^{cd}$, to obtain $F(n\prof{L},\pref^{cd})=\{d\}$ (as now $d$ gets $ns_L+x$ points).
But this violates Tops-Only (and thus Strong Contraction).
So we conclude that for all $i\in\{1,\ldots,k\}$ and all $c\in \NT(\pref)$ it is the case that $s_\pref^i(c)=k_i'\leq k_i$, with at least one $t$ such that $k'_t<k_t$.
Thus, we have that all the scoring rules award the same number of points to the top alternatives, the same number of points to the non-top alternatives, and at least one scoring rule awards a higher score to the top alternatives (thus breaking any tie among top and non-top alternatives).
This can be represented by a single scoring function $s_\pref$ such that, for all \mbox{$a\in T(\pref)$} and all $b\in \NT(\pref)$, it is the case that $s_\pref(a)=k>k'=s_\pref(b)$
By applying an additive shift, we can normalise the scoring function such that $k'=0$.

By Neutrality and Tops-Only, we only care about the number of alternatives that are in the top.
Consider two different preferences $\pref,\pref^\ast$ such that $T(\pref)\neq A $ and $|T(\pref)|>|T(\pref^\ast)|$.\footnote{Observe that, w.l.o.g., we may assume $T(\pref) \neq A$. In case $T(\pref) = A$, by applying an additive shift, we can give any score we want to the alternatives.}
Now suppose that \mbox{$k=s_\pref(a)> s_{\pref^\ast}(b)=k^\ast$ for $a\in T(\pref)$ and $b\in T(\pref^\ast)$}.
Let $b\in \NT(\pref)$, $b\in T(\pref^\ast)$, and $a\in \NT(\pref^\ast)$, $a\in T(\pref)$ (by Tops-Only and Neutrality, the structure of the top and the names of the alternatives are irrelevant).
Let $\prof{L}=(\pref_1,\pref_2)$ be a profile consisting of two linear orders with $a\pref_1 b\pref_1\ldots$ and  $b\pref_2 a\pref_2\ldots$, and all the alternatives different from $a$ and $b$ are ordered the same.
For a sufficiently large $n$ we have by Anonymity and Neutrality that $F(n\prof{L},\pref,\pref^\ast,\pref^{ab},\pref^{\ast(ab)})=\{a,b\}$.
Now remove alternatives different from $a$ from the top of $\pref$, until we get a preference $\pref'$ such that $|T(\pref')|=|T(\pref^\ast)|$.
As now $a$ gets a lower score than $b$ (as the contribution from $\pref'$ is less) we have $F(n\prof{L},\pref',\pref^\ast,\pref^{ab},\pref^{\ast(ab)})=\{b\}$, violating Strong Contraction.
Thus $k\leq k^\ast$.

If we suppose that \mbox{$k=s_\pref(a)< s_{\pref^\ast}(b)=k^\ast$ for $a\in T(\pref)$ and $b\in T(\pref^\ast)$}, and following a similar reasoning (by removing alternatives different from $b$ from the top of $\pref^{ab}$), we obtain that $k\geq k^\ast$.

We conclude that $k=k^\ast$.
By setting $k=1$, we obtain that $F$ is the Uniform Plurality rule.
\end{proof}

\noindent
Next, we proceed to establishing the independence of the axioms used in the previous results by providing examples of voting rules that satisfy all but one of the axioms featuring in our theorems.
The most interesting case is the case of a voting rule that satisfies all of the axioms but Continuity, as that axiom does not feature in the characterisation of the Plurality rule in the standard model of voting with linear preferences due to \citet{sekiguchi2012characterization}. 
Also, recall that \citet{ching1996simple} showed that Continuity can be dropped from \citeauthor{richelsonplurality1978}'s \citeyearpar{richelsonplurality1978} original characterisation. 
So the need for Continuity is a crucial difference between our Theorems~\ref{theo:plurality} and~\ref{theo:monosimpleplurality} and related results for the standard model. 
Loosely speaking, the reason for this difference is that, for our model of voting with partial orders, the different structures of the top sets permit us to vary the number of points awarded to the most preferred alternatives, making T-Congruity and Contraction less demanding.

\begin{propi}[Independence of Axioms in Theorems~\ref{theo:plurality} and~\ref{theo:monosimpleplurality}]\label{propi:indepplur}
The axioms of Anonymity, Neutrality, Reinforcement, Continuity, Partial Faithfulness, and T-Congruity/Contraction are logically independent.    
\end{propi}

\begin{proof}
So here we are claiming the independence of two sets of six axioms each.
To prove the claim, for each of the six axioms featuring in Theorem~\ref{theo:plurality} and~\ref{theo:monosimpleplurality}, we provide an example of a voting rule that satisfies all of the axioms but one, and that does not belong to the Plurality Class. 
In most cases the fact that the rules provided have the required properties is straightforward to verify; where this is not the case, we provide additional details.

\smallskip\noindent\textit{Anonymity, Neutrality, Reinforcement, Continuity, and Partial Faithfulness but neither T-Congruity nor Contraction:} the Borda-style rule that assigns to each alternative $a$ in a partial order a score equal to the number of alternatives that $a$ is preferred to.

\smallskip\noindent\textit{Anonymity, Neutrality, Reinforcement, Continuity, T-Congruity, and Contraction but not Partial Faithfulness:} the rule that always selects the full set of alternatives.

\smallskip\noindent\textit{Anonymity, Neutrality, Reinforcement, Partial Faithfulness, T-Congruity, and Contraction but not Continuity:} the rule defined by the following two-step procedure. In the first step, select all winners returned by the Uniform Plurality rule. In the second step, retain from this set only those alternatives that appear as the unique top element in at least one agent's preference---if any such alternatives exist; otherwise, keep all alternatives selected in the first step.\footnote{The inspiration for this voting rule comes from an example due to \citet{myerson1995axiomatic}.} 
See Example~\ref{examplecontinuityplur} for a failure of Continuity.
The satisfaction of the remaining properties is inherited from the Uniform Plurality rule. 
In particular, Contraction is satisfied because when a voter's top shrinks, any elected alternative that remains in the reduced top continues to be selected: either together with the other winners that were not removed, or, if it becomes the sole top element, as the unique winner. 

\smallskip\noindent\textit{Anonymity, Neutrality, Continuity, Partial Faithfulness, T-Congruity, and Contraction but not Reinforcement:} the rule such that for single-voter profiles selects the top of the sole voter and for other profiles selects the Uniform Plurality winners and all the alternatives that get exactly one point less than the Uniform Plurality winners. 
The following example shows that this rule indeed violates Reinforcement. 
If in profile $\prof{\pref}$ the alternatives $a$ and $b$ appear 10 and 9 times in the top, respectively, and in $\prof{\pref}'$ they appear 7 and 6 times, respectively, then $\{a\}=F(\prof{\pref},\prof{\pref}')\neq F(\prof{\pref})\cap  F(\prof{\pref}')=\{a,b\}$.

\smallskip\noindent\textit{Anonymity, Reinforcement, Continuity, Partial Faithfulness, T-Congruity, and Contraction but not Neutrality:} the rule assigning points to alternatives just like the Uniform Plurality rule but that doubles the score of one fixed alternative $a\in A$ before determining the winners.

\smallskip\noindent\textit{Neutrality, Reinforcement, Continuity, Partial Faithfulness, T-Congruity, and Contraction but not Anonymity:} the rule assigning points to alternatives just like the Uniform Plurality rule except for voter~$1$'s top alternatives, which each receive $2$ points.
\end{proof}
\begin{exa}\label{examplecontinuityplur}
To exemplify the operation of the rule demonstrating that Continuity is a necessary axiom for our characterisation of the Plurality Class, consider the following two preferences, $\pref_1$ and $\pref_2$:
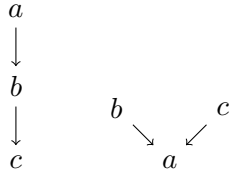
\begin{center}\begin{tikzpicture} 
\node (1) {$c$};
\node (3) [above of=1] {$b$};
\node (2) [above of=3] {$a$};
\draw[->] (2) -- (3);
\draw[->] (3) -- (1);
\end{tikzpicture}
\qquad
\begin{tikzpicture}
\node (1) {$a$};
\node (2) [above left of=1] {$b$};
\node (3) [above right of=1] {$c$};
\draw[->] (2) -- (1);
\draw[->] (3) -- (1);
\end{tikzpicture}\end{center}
In the two-agent profile $(\pref_1,\pref_2)$, in the first step we select $\{a,b,c\}$, but as only $a$ appears as a unique top, in the second step we obtain $\{a\}$.

The failure of Continuity can now be observed by noting that the rule instead returns $\{b,c\}$ for any profile that is composed of $k>1$ copies of $(\pref_1,\pref_2)$ and one further copy of $(\pref_2)$. 
Indeed, for any such profile, we select $\{b,c\}$ in the first step, after which no further selection will occur in the second step.
\end{exa}

\noindent
The next statement, regarding the independence of the axioms of Theorem~\ref{theo:uniformplurality}, is presented without proof, as the same examples as in the proof of Proposition~\ref{propi:indepplur} can be used (except that certain rules also satisfying Continuity are not relevant here).

\begin{propi}[Independence of Axioms in Theorem~\ref{theo:uniformplurality}]\label{propi:indepunifplur}
The axioms of Anonymity, Neutrality, Reinforcement, Partial Faithfulness, and Strong Contraction are logically independent.    
\end{propi}

\subsection{The Anti-Plurality Class}\label{sec:antipluraclass}

In this section, we present the characterisation of the voting rules that make up the Anti-Plurality Class.\footnote{Some of the proofs are similar to those for the corresponding results regarding the Plurality Class, although it is not the case that we can obtain the results for the Anti-Plurality Class as corollaries to those for the Plurality Class.}

\begin{teo}\label{theo:antiplurality}
A voting rule satisfies Anonymity, Neutrality, Reinforcement, Continuity, Partial Averseness, and B-Congruity if, and only if, it belongs to the Anti-Plurality class.       
\end{teo}

\begin{proof}
By Theorem~\ref{thm:psr}, we know that $F$ is a positional scoring rule.
Let ${\pref}\in\mathcal{D}$ and suppose there are alternatives $a\in \NB(\pref)$ and $c\in B(\pref)$ such that $s_\pref(a)<s_\pref(c)$.
We consider the profile $\overline{\prof{\pref}}$, and by Anonymity and Neutrality, we have that $F(\overline{\prof{\pref}})=A$.
Thus $a\notin F(\overline{\prof{\pref}},\pref)$, contradicting B-Congruity.
Then $s_\pref(a)\geq s_\pref(c)$ for all $a\in \NB(\pref)$ and $c\in B(\pref)$.
Using a similar argument we can show that there exists a $k$ such that $s_\pref(a)=s_\pref(d)=k$ for all $a, d\in \NB(\pref)$. 
By Partial Averseness, there is at least one alternative $b\in B(\pref)$ that gets a score $s_\pref(b)<k$, concluding that $F$ belongs to the Anti-Plurality class.
\end{proof}

\noindent
Averseness forces the scoring rule to give, within a preference, the same number of points to all the bottom alternatives.

\begin{coro}\label{coro:simpleantiplurality}
A voting rule satisfies Anonymity, Neutrality, Reinforcement, Continuity, Averseness, and B-Congruity if, and only if, it belongs to the Simple Anti-Plurality class.      
\end{coro}

\noindent
Recall that Expansion requires the scoring rule to assign more points to bottom alternatives when they occur together with other bottom alternatives.\footnote{Theorem~\ref{theo:monosimpleantiplurality} continues to holds if we relax Averseness to Partial Averseness.}

\begin{teo}\label{theo:monosimpleantiplurality}
A voting rule satisfies Anonymity, Neutrality, Reinforcement, Continuity, Averseness, and Expansion if, and only if, it belongs to the Monotonic Simple Anti-Plurality class.     
\end{teo}

\begin{proof}
By Theorem~\ref{thm:psr}, we know that $F$ is a positional scoring rule.
Let ${\pref_i}\in\mathcal{D}$ and assume $|B(\pref)|\geq 2$.
Suppose there are two alternatives $b,c\in B(\pref)$ such that $s_{\pref}(b)<s_{\pref}(c)$.
When we consider the profile $\overline{\prof{\pref}}$, we know that $F(\overline{\prof{\pref}})=A$.
Consider $\pref_i^{bc}$ that is $\pref$ with $b$ and $c$ are relabelled.
Thus \mbox{$F(\pref_i^{bc},\overline{\prof{\pref}}_{-i})=\{b\}$}.
But $B(\pref)\subseteq B(\pref_i^{bc})$ and , so we have, by Expansion, that $B(\pref)\cap F(\overline{\prof{\pref}})\subseteq (\pref_i^{bc},\overline{\prof{\pref}}_{-i})$, obtaining a contradiction.
We conclude that there is a $k$ such that $s_{\pref}(b)=k$ for all $b\in B(\pref)$.

Now assume that $|\NB(\pref)|\geq 2$ and that $s_\pref(a)\neq s_\pref(b)$ for some alternatives $a,b\in \NB(\pref)$.
Using a similar argument as before, we obtain that there exists a $k'$ such that $s_\pref(a)=k'$ for all $a\in \NB(\pref)$.
Finally, by Averseness, we have that $k<k'$, concluding that $F$ belongs to the Simple Anti-Plurality class.

As Expansion only considers the alternatives in the bottom of a preference and not the structure, and by Neutrality, we have that for any two preferences $\pref$ and $\pref'$ such that $|B(\pref)|=|B(\pref')|$, it is the case that $s_\pref(a)=s_{\pref'}(b)$ for any $a\in B(\pref)$ and $b\in B(\pref')$, or any $a\in \NB(\pref)$ and $b\in \NB(\pref')$.

We normalise the scores by means of additive shifts, such that for all ${\pref}\in\mathcal{D}$ we get $s_\pref(a)=0$ for all $a\in \NB(\pref)$.
Now consider two preferences ${\pref},{\pref'}\in\mathcal{D}$ where \mbox{$|B(\pref)|<|B(\pref')|$}, and $k$ and $k'$ are the scores for the bottom alternatives of $\pref$ and $\pref'$, respectively. Suppose $k>k'$.
By Neutrality, and w.l.o.g., we can assume that $B(\pref)\subset B(\pref')$.
Again, we consider the profile $\overline{\prof{\pref}}$ and find that $F(\overline{\prof{\pref}})=A$.
Now we assume that agent $i$ has preference $\pref$.
We have that $B(\pref)\cap F(\overline{\prof{\pref}})=B(\pref)$ and thus, by Expansion, $B(\pref)\subseteq F(\pref',\overline{\prof{\pref}}_{-i})$.
But in $F(\pref',\overline{\prof{\pref}}_{-i})$ all the alternatives in $B(\pref')$ get their scores reduced, while those in $\NB(\pref')$ remain with the same score, obtaining $F(\pref',\overline{\prof{\pref}}_{-i})=A\setminus{B(\pref')}$.
With this contradiction, we can conclude that $k\leq k'$.
Thus, we obtain that $F$ belongs to the Monotonic Simple Anti-Plurality class.
\end{proof}

\noindent
We now introduce another technical property: Bottoms-Only. 
It requires that only the least-preferred alternatives of the voters influence the outcome.
Formally, a rule satisfies Bottoms-Only if for any two profiles ${\prof{\pref}},{\prof{\pref}'}\in\mathcal{D}^N$ with $B(\pref_i)=B(\pref_i')$ for all $i\in N$, it is the case that $F(\prof{\pref})=F(\prof{\pref}')$.\footnote{For a possible motivation regarding this property, and its use for a characterisation result, we refer to \citet{KURIHARA2018110}.}
It turns out that any rule satisfying Strong Expansion depends only on the voters’ bottom-ranked alternatives. 
We omit the proof, which is analogous to that of Lemma~\ref{lema:coherencetopsonly}.

\begin{lema}\label{lema:strongexpimpliesbo}
Strong Expansion implies Bottoms-Only.    
\end{lema}
   
\noindent
Here again, to characterise the Uniform Anti-Plurality rule, we build on a characterisation of the Anti-Plurality rule for linear orders by \citet{KURIHARA2018110} in terms of Anonymity, Neutrality, Reinforcement, Partial Averseness, and Bottoms-Only; and the aforementioned result by \citet{PIVATO2013210}.

\begin{teo}\label{theo:uniformantiplurality}
A voting rule satisfies Anonymity, Neutrality, Reinforcement, Partial Averseness, and Strong Expansion if, and only if, it is the Uniform Anti-Plurality rule.    
\end{teo}

\begin{proof}
By Lemma~\ref{lema:strongexpimpliesbo}, we have that any rule that satisfies our axioms also satisfies Bottoms-Only.
By the result of \citet{KURIHARA2018110}, we know that when the profile consists of linear orders, a voting rule satisfies all the axioms if, and only if, it is the Anti-Plurality rule.
By an additive shift, we can normalise the rule such that the score for the non-bottom alternatives of a linear order is $0$, and let $s_L<0$ be the score for the bottom alternative.

By the result due to \citet{PIVATO2013210} mentioned earlier, we know that a rule satisfies Anonymity, Neutrality and Reinforcement if, and only if, it is a composite scoring rule. 
Thus, our rule involves a sequence of scoring functions $s^1,\ldots,s^k$ where $s^{k'}$ is used to break ties not yet broken by $s_1,\ldots,s_{k'-1}$.

Let ${\pref}\in\mathcal{D}$ be a generic preference.
By Neutrality, Partial Averseness and Bottoms Only, we have that $F(\pref)=\NB(\pref)$ (to see this, consider the preference where all the non-bottom and all the bottom alternatives are symmetric).
This implies that for all $i\in\{1,\ldots,k\}$, all \mbox{$a\in \NB(\pref)$}, and all $b\in B(\pref)$, $s_\pref^i(a)=k_i$ and $k_i\geq s_\pref^i(b)$, with at least one $t$ such that $k_t> s_\pref^t(b)$.
Now we assume there is a $ s_\pref^j$ and two positions at the bottom such that alternatives $c,d\in B(\pref)$ get scores $x=s_\pref^j(c)>s_\pref^j(d)=y$ (if there are more alternatives with different scores, we select those two with the highest scores).
Let $\prof{L^c}=(L_1^c,\ldots,L_{m-2}^c)$ be a profile of $m-2$ linear orders such that $c$ is always on top, $d$ is always second, and the rest of the alternatives appear exactly once at the bottom.
Now let $\prof{L^d}=(L_1^d,\ldots,L_{m-2}^d)$ be the same as $\prof{L^c}$ but with $c$ and $d$ swapped.
As for linear orders $F$ acts as the Anti-Plurality rule, we obtain that $F(\prof{L^c},\prof{L^d})=\{c,d\}$, where $c$ and $d$ get 0 points and all other alternatives get $2s_L$.
Let $n$ be such that $s_\pref^j(c)  >  2ns_L+\max_i k_i$.
Thus $F(\pref,n(\prof{L^c},\prof{L^d}))=\{c\}$.
Now consider the profile where $c$ and $d$ are relabelled.
Then, by Neutrality, we have that $F(\pref^{cd},n(\prof{L^d},\prof{L^c}))=\{d\}$.
But this violates Bottoms-Only (and thus Strong Expansion).
So we conclude that for all $i\in\{1,\ldots,k\}$ and all $c\in B(\pref)$ it is the case that $s_\pref^i(c)=k_i'\leq k_i$, with at least one $t$ such that $k'_t<k_t$.
This can be represented by a single scoring function $s_\pref$ such that for all \mbox{$a\in \NB(\pref)$} and all $b\in B(\pref)$, it is the case that $s_\pref(a)=k>k'=s_\pref(b)$
By an additive shift we can normalise the score such that for every preference $\pref$ we get $s_\pref(a)=0$ for all $a\in \NB(\pref)$.

The remainder of the proof now proceeds in an analogous manner to the proof of Theorem~\ref{theo:uniformplurality}.
\end{proof}

\noindent
Next, as we previously did for the results of Section~\ref{sec:pluraclass}, we show that the axioms used for the results in Section~\ref{sec:antipluraclass} are logically independent.
Again, it is interesting to observe that Continuity is required for Theorems~\ref{theo:antiplurality} and~\ref{theo:monosimpleantiplurality}, in contrast to the result by \citet{KURIHARA2018110} for the case of linear orders.
This is so because of the different structures that the bottom set can have, making B-Congruity and Expansion less demanding.

\begin{propi}[Independence of Axioms in Theorems~\ref{theo:antiplurality} and~\ref{theo:monosimpleantiplurality}]\label{propi:indepantiplur}
The axioms of Anonymity, Neutrality, Reinforcement, Continuity, Partial Averseness, and B-Congruity/Expansion are logically independent.    
\end{propi}

\begin{proof}
We demonstrate the existence of a relevant rule for each case.

\smallskip\noindent\textit{Anonymity, Neutrality, Reinforcement, Continuity, and Partial Averseness but neither B-Congruity nor Expansion:} the Borda-style rule also used in the proof of Proposition~\ref{propi:indepplur}.

\smallskip\noindent\textit{Anonymity, Neutrality, Reinforcement, Continuity, B-Congruity, and Expansion but not Partial Averseness:} the rule that always selects the full set of alternatives.

\smallskip\noindent\textit{Anonymity, Neutrality, Reinforcement, Partial Averseness, B-Congruity, and Expansion but not Continuity:} the rule defined by means of the following two-step procedure.
In the first step, select all the winning alternatives returned by the Uniform Anti-Plurality rule.
In the second step, select as winning alternatives from this set those alternatives that do not appear as a unique bottom element in at least one of the agents' preferences---in case there are such alternatives; otherwise, pick all alternatives selected in the first step.
Similar arguments to the ones used for the voting rule which violates continuity in Proposition~\ref{propi:indepplur} hold for this example, as well.

\smallskip\noindent\textit{Anonymity, Neutrality, Continuity, Partial Averseness, B-Congruity, and Expansion but not Reinforcement:} the rule that for a single-voter profile selects all the non-bottom alternatives and for other profiles selects the Uniform Anti-Plurality winners and all the alternatives that get exactly one point less than the Uniform Anti-Plurality winners.

\smallskip\noindent\textit{Anonymity, Continuity, Reinforcement, Partial Averseness, B-Congruity, and Expansion but not Neutrality:} the rule that assigns points to alternatives just like the Uniform Anti-Plurality rule but that doubles the score of one fixed alternative $a\in A$ before determining the winners.

\smallskip\noindent\textit{Neutrality, Reinforcement, Continuity, Partial Averseness, B-Congruity, and Expansion but not Anonymity:} the rule that assigns points to alternatives just like the Uniform Anti-Plurality rule, except for voter~$1$'s bottom alternatives, which receive $-2$ points.
\end{proof}

\noindent
The next statement, regarding the independence of the axioms featuring in Theorem~\ref{theo:uniformantiplurality},  is presented without proof, as the same examples of Proposition~\ref{propi:indepantiplur} can be used. 

\begin{propi}[Independence of Axioms in Theorems~\ref{theo:uniformantiplurality}]\label{propi:indepunifanti}
The axioms of Anonymity, Neutrality, Reinforcement, Partial Averseness, and Strong Expansion are logically independent.    
\end{propi}

\subsection{Applications to Approval Voting}\label{sec:appvot}

In the approval voting model, voters submit a ballot indicating the alternatives they approve of, and it is typically assumed that any alternatives not included in the ballot are disapproved of by the voter \citep{bramsfishburnapporvalvoting1978}.
This type of ballot can be simulated in our setting as a ballot where every alternative is either in the top or in the bottom of the preference.
Formally, in our model, a preference ${\pref}\in\mathcal{D}$ is an approval ballot if $T(\pref)\cap B(\pref)=\emptyset$ and $T(\pref)\cup B(\pref)=A$.

It is not hard to see that for approval ballots, Averseness becomes equivalent to Partial Faithfulness, and Strong Expansion becomes equivalent to Strong Contraction.
In our setting, the \emph{size approval voting rules} characterised by \citet{ALCALDEUNZU20091187} and \citet{BardalEndrissGEB2025} are equivalent to the rules of the Monotonic Simple Plurality Class; while the standard approval voting rule \citep[see][for several characterisations]{BrandlPetersJET2022} is equivalent, for approval ballots, to both the Uniform Plurality and Uniform Anti-Plurality rule.
These observations lead to the following new characterisations, which, to the best of our knowledge, are independent of previously known ones \citep{ALCALDEUNZU20091187,BrandlPetersJET2022,BardalEndrissGEB2025}.

The next result is implied by the variant of Theorem~\ref{theo:monosimpleplurality} (briefly mentioned earlier) where we substitute Faithfulness by Partial Faithfulness.

\begin{coro}\label{coro:sizeapppfcontr}
A voting rule for approval ballots
satisfies Anonymity, Neutrality, Reinforcement, Continuity, Partial Faithfulness, and Contraction if, and only if, it is a size approval voting rule.
\end{coro}

\noindent
Note that \citet{BardalEndrissGEB2025} state a similar result, but using what they call Weak Faithfulness rather than our Partial Faithfulness (thus requiring $T(\pref)\subseteq F(\pref)$ instead of $F(\pref)\subseteq T(\pref)$).

The following two results are implied by Theorems~\ref{theo:uniformplurality} and~\ref{theo:uniformantiplurality}, respectively.
\begin{coro}\label{coro:standapppfsc}
A voting rule for approval ballots
satisfies Anonymity, Neutrality, Reinforcement, Partial Faithfulness, and Strong Contraction if, and only if, it is the standard approval voting rule.
\end{coro}

\begin{coro}\label{coro:standapppasc}
A voting rule for approval ballots
satisfies Anonymity, Neutrality, Reinforcement, Partial Averseness, and Strong Contraction if, and only if, it is the standard approval voting rule.
\end{coro}

\noindent
We can also consider an approval voting rule that works in the `opposite' direction of size approval voting.
Under such a rule, an approved alternative receives \emph{more} points when it is approved together with a larger number of other alternatives (while the opposite is true for size approval voting rules).
Such rules favour voters who are willing to compromise and approve more alternatives \citep[see also Example~6 in the work of][]{BardalEndrissGEB2025}.
Let us call this family of voting rules the \emph{anti-size approval rules}.
We can then state the following result implied by Theorem~\ref{theo:monosimpleantiplurality}.

\begin{coro}\label{coro:antisizeapppae}
A voting rule for approval ballots
satisfies Anonymity, Neutrality, Reinforcement, Continuity, Partial Averseness, and Expansion if, and only if, it is an anti-size approval voting rule.
\end{coro}

\noindent
We emphasise that all these characterisation results are derived without directly relying on the technical assumption of Tops-Only (although that axiom is of course implied by Strong Contraction), which is implicit in the approval voting model (as approval ballots only state the most preferred alternatives of a voter).

%%%%%%%%%%%%%%%%%%%%%%%%%%%%%%%%%%%%%%%%%%%%%%%%%%%%%%%%%%%%%%%%%%%%%%%%%%%%%%%%
\section{Conclusion}\label{sec:conclusion}
%%%%%%%%%%%%%%%%%%%%%%%%%%%%%%%%%%%%%%%%%%%%%%%%%%%%%%%%%%%%%%%%%%%%%%%%%%%%%%%%

We have provided axiomatic characterisations of generalised versions of the Plurality and Anti-Plurality rules for settings in which preferences are represented by partial orders. 
In the classical setting, where preferences are linear orders, the Plurality (Anti-Plurality) rule assigns one point (one negative point) to the most (least) preferred alternative and zero points to all others. 
We have extended this defining feature to the case of partial orders in a way that respects the structural richness of such preferences, for instance, by recognising that two top alternatives may differ in the number of alternatives they dominate, and thus may not be equally `top-ranked.'

When considering linear orders, the characterisations we have presented all converge to familiar standards established in the existing literature. But when considered in the context of voting with partial-order preferences, they illustrate some of the subtleties of this important model of voting, providing voters with additional means of expressing themselves.  

We also have demonstrated that several of our findings carry over to the approval voting setting, enabling us to derive new axiomatic characterisations of some well-known as well as new approval voting rules.

A summary of all our characterisation results can be found in Table~\ref{tab:results}. In this table,  simple check mark~(\yes) indicates that the axiom in question is part of the characterisation of the class or rule in question, while a check mark in parentheses indicates that an alternative characterisation in terms of that axiom is possible as well (the latter concerns two simple variants of our results we have stated in footnotes without full proof).

An intriguing path for future investigation lies in exploring similar generalisations of other scoring rules to the domain of partial orders, including but not limited to the Cumulative and Stepwise scoring rules.\footnote{See \citet{kruger2020strategic} for definitions of these scoring rules.}

%%%%%%%%%%%%%%%%%%%%%%%%%%%%%%%%%%%%%%%%%%%%%%%%%%%%%%%%%%%%%%%%%%%%%%%%%%%%%%%%

\begin{table}[ht]
\centering
\caption{Summary of characterisation results}
\label{tab:results}
\renewcommand{\arraystretch}{1.25}
\small
\resizebox{\textwidth}{!}{%
\setlength{\tabcolsep}{4pt}
\begin{tabular}{%
  >{\raggedright\arraybackslash}p{8cm} 
  c c c c  
  c c      
  c c      
  c c      
  c c      
  c c      
}
\toprule
\textbf{Result}
  & \rotatebox{90}{\emph{Anon.}}
  & \rotatebox{90}{\emph{Neutr.}}
  & \rotatebox{90}{\emph{Reinf.}}
  & \rotatebox{90}{\emph{Cont.}}
  & \rotatebox{90}{\emph{P.F.}}
  & \rotatebox{90}{\emph{F.}}
  & \rotatebox{90}{\emph{P.Av.}}
  & \rotatebox{90}{\emph{Av.}}
  & \rotatebox{90}{\emph{T-C.}}
  & \rotatebox{90}{\emph{B-C.}}
  & \rotatebox{90}{\emph{Contr.}}
  & \rotatebox{90}{\emph{S.C.}}
  & \rotatebox{90}{\emph{Exp.}}
  & \rotatebox{90}{\emph{S.E.}} \\

\midrule
\multicolumn{15}{l}{\textit{Plurality family}} \\[1pt]

\rowcolor{blue!12}
Plurality Class (Thm.~\ref{theo:plurality})
  & \yes & \yes & \yes & \yes
  & \yes &  &  &  & \yes &  &  &  &  &
  \\

\rowcolor{blue!12}
Simple Plurality Class (Cor.~\ref{coro:simpleplurality})
  & \yes & \yes & \yes & \yes
  &  & \yes &  &  & \yes &  &  &  &  &
  \\

\rowcolor{blue!12}
Monotonic Simple Plurality Class (Thm.~\ref{theo:monosimpleplurality})
  & \yes & \yes & \yes & \yes
  & (\yes) & \yes  &  &  &  &  & \yes &  &  &
  \\

\rowcolor{blue!12}
Uniform Plurality rule (Thm.~\ref{theo:uniformplurality})
  & \yes & \yes & \yes &
  & \yes &  &  &  &  &  &  & \yes &  &
  \\[2pt]

\midrule
\multicolumn{15}{l}{\textit{Anti-Plurality family}} \\[1pt]

\rowcolor{green!12}
Anti-Plurality Class (Thm.~\ref{theo:antiplurality})
  & \yes & \yes & \yes & \yes
  &  &  & \yes &  &  & \yes &  &  &  &
  \\

\rowcolor{green!12}
Simple Anti-Plurality Class (Cor.~\ref{coro:simpleantiplurality})
  & \yes & \yes & \yes & \yes
  &  &  &  & \yes &  & \yes &  &  &  &
  \\

\rowcolor{green!12}
Monotonic Simple Anti-Plurality Class (Thm.~\ref{theo:monosimpleantiplurality})
  & \yes & \yes & \yes & \yes
  &  &  & (\yes) & \yes  &  &  &  &  & \yes &
  \\

\rowcolor{green!12}
Uniform Anti-Plurality rule (Thm.~\ref{theo:uniformantiplurality})
  & \yes & \yes & \yes &
  &  &  & \yes &  &  &  &  &  &  & \yes
  \\[2pt]

\midrule
\multicolumn{15}{l}{\textit{Approval voting}} \\[1pt]

\rowcolor{gray!10}
Size approval voting rules (Cor.~\ref{coro:sizeapppfcontr})
  & \yes & \yes & \yes & \yes
  & \yes &  &  &  &  &  & \yes &  &  &
  \\

\rowcolor{gray!10}
Standard approval voting (Cor.~\ref{coro:standapppfsc})
  & \yes & \yes & \yes &
  & \yes &  &  &  &  &  &  & \yes &  &
  \\

\rowcolor{gray!10}
Standard approval voting (Cor.~\ref{coro:standapppasc})
  & \yes & \yes & \yes &
  &  &  & \yes &  &  &  &  & \yes &  &
  \\

\rowcolor{gray!10}
Anti-size approval voting rules (Cor.~\ref{coro:antisizeapppae})
  & \yes & \yes & \yes & \yes
  &  &  & \yes &  &  &  &  &  & \yes &
  \\

\bottomrule
\end{tabular}%
}
\end{table}

\paragraph{Acknowledgements.}
We are grateful to Tuva Bardal, Jordi Mass\'o, Fernando Tohm\'e, Philippe Solal, several attendees of the COMSOC Seminar at the University of Amsterdam and the Conference on Mechanism and Institution Design held in Budapest in 2024, as well as several anonymous reviewers for comments and suggestions that led to improvements of this paper.
We acknowledge the financial support by the Dutch Research Council (NWO) through the VICI scheme (grant number 639.023.811) and by the French National Research Agency (ANR) through project ANR-24-EXMA-0001 PEPR MathsVives CONDORCET.

%\paragraph{Statements and Declarations.}
%The authors have no competing interests to declare that are relevant to the content of this article.

\bibliography{ref}
\end{document}